\documentclass[aps,prd,onecolumn,showpacs,preprintnumbers,nofootinbib,12pt]{revtex4-2}

\usepackage[utf8]{inputenc}      
\usepackage[T1]{fontenc}         
\usepackage[english]{babel}
\usepackage{amsmath,amssymb,amsfonts,amsthm}
\usepackage{graphicx}            
\usepackage{hyperref}            
\usepackage{tensor}              
\usepackage{mathrsfs}            
\usepackage{bm}                  
\usepackage{enumerate}           
\usepackage{physics}             
\usepackage{float}               
\usepackage{url}                 
\usepackage{mathtools}
\usepackage{newtxtext,newtxmath}

\hypersetup{
	colorlinks = true,
	linkcolor  = blue,
	citecolor  = blue,
	urlcolor   = blue,
	pdfauthor  = {J. A. Silva, F. C. Carvalho, A. R. G. Garcia},
	pdftitle   = {Generalized junction conditions for discontinuous metrics}
}

\newtheorem{theorem}{Theorem}

\newtheorem{definition}{Definition}
\newtheorem{lemma}{Lemma}

\begin{document}

\title{Generalized junction conditions for discontinuous metrics}

\author{J.~A.~Silva}
\email{jonatassilva@alu.uern.br}
\author{F.~C.~Carvalho}
\email{fabiocabral@uern.br}
\affiliation{Departamento de Física, Universidade do Estado do Rio Grande do Norte (UERN), Mossoró, RN, Brasil}

\author{Antonio R.~G.~Garcia}
\email{ronaldogarcia@ufersa.edu.br}
\affiliation{Departamento de Ciências Naturais, Matemática e Estatística, Universidade Federal Rural do Semi-Árido (UFERSA), Av.~Francisco~Mota, Mossoró, 59625-900, RN, Brasil}

\date{\today}

\begin{abstract}
	In this work, the Darmois–Israel junction formalism is extended to the case of discontinuous metrics within the framework of Colombeau algebras of generalized functions. This formulation provides a mathematically consistent treatment of nonlinear operations involving singular quantities, such as products and derivatives of distributions. By relaxing the usual continuity condition on the metric, the generalized junction conditions naturally include higher-order singular terms in the curvature and in the surface energy–momentum tensor. These additional contributions represent new geometric degrees of freedom associated with genuine discontinuities in the space–time geometry. The resulting formalism recovers the traditional Darmois–Israel conditions as a limiting case, while offering a coherent extension applicable to geometric boundaries and abrupt transitions in space–time.
	
	\vspace{0.5em}
\end{abstract}

\maketitle

\section{Introduction}

The formulation of junction conditions plays a fundamental role in the description of geometric boundaries in spacetime, where different regions may possess distinct material contents or metric properties. Since the early days of general relativity, this problem raises a central question: how to describe, in a mathematically consistent way, abrupt transitions between matter and vacuum, or between two incompatible geometries, while ensuring the physical coherence of Einstein’s equations? \cite{Raju1982} The difficulty arises from the fact that Einstein’s equations are intrinsically nonlinear, which prevents the direct application of the linear theory of distributions to singular or discontinuous metrics. Thus, the formulation of junction conditions must reconcile the physical aspect, the characterization of infinitesimal boundaries of energy and curvature, with the mathematical consistency required to handle quantities that involve products of distributions.

This question presented above involves a delicate point: distinguishing between real discontinuities in the geometry of spacetime and those that arise only from poorly chosen coordinates. The formulation of this problem was initially discussed by Lanczos (1922, 1924) \cite{Lanczos1922, Lanczos1924}, and, in the following decades, received several reformulations. Works such as those of Raychaudhuri (1953) \cite{Raychaudhuri1953} and Israel (1958) \cite{Israel1958} recognized the importance of properly treating geometric discontinuities, but still without a clear distinction between coordinate effects and genuine physical discontinuities.

The first invariant treatment was proposed by O’Brien and Synge (1952) \cite{OBrienSynge}, followed by Lichnerowicz (1955) \cite{Lichnerowicz}, who introduced the concept of admissible coordinates, requiring a \(\mathscr{C}^2\) atlas in which the components \( g_{ab} \) were \(\mathscr{C}^3\), except on smooth hypersurfaces. Synge (1960) \cite{Synge1960} refined this formulation, and Israel (1966) \cite{Israel1966} presented the most influential generalization, treating singular hypersurfaces with different extrinsic curvatures and introducing surface energy as a projection of material energy. In parallel, Dautcourt (1964) \cite{Dautcourt1964}, Papapetrou and collaborators \cite{Papapetrou1959, Papapetrou1968} developed complementary approaches for singular surfaces.

These milestones consolidated the foundations of the modern formulation of junction conditions, which are crucial for modeling astrophysical systems with thin boundaries, such as matter shells, gravitational shock waves, stellar collapse, and screening effects. The regularity of the components \( g_{ab} \), especially around hypersurfaces, proves to be both a mathematical and physical requirement \cite{Raju1982}.

One of the first systematic formulations for the transition between regions of spacetime was given by Darmois (1927) \cite{Darmois1927}, by establishing conditions of continuity of the metric and of its normal derivatives on non-null hypersurfaces, based on the Cauchy problem for Einstein’s equations. Although restricted by assumptions of regularity and symmetry, this treatment consolidated the so-called Darmois conditions, fundamental for understanding how regions of spacetime can be smoothly joined. Decades later, Choquet-Bruhat (1968) \cite{ChoquetBruhat1968} and again Israel (1966) \cite{Israel1966} substantially expanded this framework. Using Gaussian coordinates along \(\Sigma\), Israel showed that although the normal derivatives of the metric may be discontinuous, the components \( g_{ab} \) remain continuous, provided that both regions induce the same intrinsic metric on \(\Sigma\). In natural coordinates, this is equivalent to requiring that \( g_{ab} \) be of class \(\mathscr{C}^0\) on the singular hypersurface.

Applying Einstein’s equations on both sides of the junction, Israel deduced a three-dimensional conservation law for surface layers, establishing the foundations of the Darmois–Israel formulation, now known as the thin-shell formalism. This approach provides coherent gluing conditions between distinct regions of spacetime, based on the continuity of the metric and the controlled discontinuity of the extrinsic curvature. Besides possessing a clear physical interpretation, the equality of the intrinsic geometry on both sides of \(\Sigma\) avoids mathematical difficulties associated with ill-defined products of distributions.

This formalism has become a widely used tool in general relativity, with applications ranging from thin-shell wormholes \cite{Visser1995}, gravastars \cite{Mazur2023, Uchikata2016}, and stellar models \cite{Fayos1992, Fayos1996}, to cosmological domain walls \cite{Blau1987}. It also underlies the cut-and-paste techniques, in which the hypersurface \(\Sigma\) acts as a geometric boundary between two regions of spacetime.

Although the traditional derivation of the Darmois–Israel conditions often relies on the Gauss–Codazzi equations \cite{Israel1966, Mansouri1996, Misner}, in this work an alternative approach is adopted, following Padmanabhan \cite{Padmanabhan2010} (see also Poisson \cite{Poisson2004}). Instead of resorting to geometric identities, one starts from the representation of the metric in terms of the Heaviside function and examines the singular terms in Einstein’s equations, associated with the derivative of the metric containing delta-type distributions. In this way, the standard Darmois–Israel conditions arise naturally from the continuity of the induced metric and the discontinuity of the extrinsic curvature tensor, related to the surface energy–momentum tensor.

In this context, it is common to model spacetime as the union of two regions \( M^+ \) and \( M^- \) glued along a hypersurface \(\Sigma\), with global metric
\begin{equation}\label{metricadiscont}
g_{ab} = \theta(l)\, g^+_{ab} + \theta(-l)\, g^-_{ab}~,
\end{equation}
where \(l\) is a smooth function that vanishes on \(\Sigma\), and \(\theta(l)\) is the Heaviside function. This construction makes it possible to describe abrupt junctions, but it presupposes the continuity of the induced metric, that is, \( [h_{\alpha\beta}] = 0 \).

However, it is important to emphasize that there are fundamental limitations in the distributional formulation of Einstein’s equations. Geroch and Traschen (1987) \cite{GerochTraschen1987} showed that the nonlinear nature of the curvature tensor generally prevents the treatment of highly singular metrics within the classical theory of Schwartz distributions. Although weak derivatives of the metric can be defined, the curvature involves products of these distributions, operations that are not mathematically well defined in the linear framework. To circumvent this problem, the authors introduced the notion of a regular metric, that is, a symmetric tensor \( g_{ab} \) whose inverse \( g^{ab} \) exists and is locally bounded, and whose weak derivative is locally square-integrable. Under these conditions, the Einstein tensor can be interpreted in a well-defined distributional sense. This restriction, however, excludes metrics with more severe discontinuities (such as those containing delta terms in the metric itself), revealing an intrinsic limitation of the standard distributional theory applied to general relativity.

Although Racskó (2024) \cite{Racsko2024} highlights the inherent difficulties in products between singular distributions and discontinuous functions, which often arise in junction problems and do not admit a consistent definition in the classical formalism, he shows that common approaches—such as the ad hoc interpretation based on averaging one-sided limits—may lead to incorrect results. His formalism, still based on the traditional theory of distributions, avoids such “illegitimate’’ operations by imposing appropriate regularity conditions on the field variables, thereby circumventing the occurrence of undefined products. This strategy, consistent with Schwartz’s impossibility theorem, ensures mathematical consistency but restricts the analysis to situations in which more intense singularities, such as \( \delta^2 \), do not appear explicitly.

In higher-order derivative gravitational theories, this assumption becomes restrictive. The appearance of singular terms such as \( \delta' \) or \( \delta^2 \) prevents a rigorous treatment within the classical theory of distributions. Reina et al. (2016) discuss these difficulties when studying junctions in quadratic gravity, recognizing the relevance of Colombeau algebras to overcome such limitations, although without implementing them \cite{Reina2016}. Chu (2022) \cite{ChuTan2022} introduced a technique of regularized integration of the field equations, robust enough to handle singular distributions in \( \mathcal{F}(R) \) theories. More recently, Huber (2025) \cite{Huber2025} explicitly implemented the Colombeau formalism, allowing one to handle distributional metrics, compute curvature invariants in higher-derivative theories, and establish connections with Penrose’s cut-and-paste method.

Despite these advances, and previous proposals by Huber himself (2020) \cite{Huber2020} for local deformations of metrics, it appears that a systematic formulation has yet to be developed that explicitly and rigorously treats the junction conditions for discontinuous metrics, with \( [g_{ab}] \neq 0 \) or \( [h_{\alpha\beta}] \neq 0 \), resolving the distributional ambiguities that arise in this regime.

In this work, we propose precisely this generalization. We relax the condition of continuity of the induced metric and consider junctions with genuine geometric breaking along \(\Sigma\). To this end, we refine the traditional methods of the thin-shell formalism by employing the framework of Colombeau algebras, which allows one to manipulate in a well-defined manner products and derivatives of distributions such as \( \theta \delta \), \( \theta \delta' \), and \( \delta^2 \).

As a result, we obtain an extension of the Darmois–Israel formalism applicable to situations in which the metric exhibits genuine discontinuities, broadening the physical scope of the traditional formalism and providing a consistent mathematical framework for the study of non-smooth junctions in general relativity and modified gravitational theories.

The conventions adopted throughout this work are \( c = G = 1 \) and metric signature \((- , + , + , +)\).
Brackets \( [F] \) denote the jump of a quantity \( F \) across the hypersurface \( \Sigma \), that is, \( [F] = F^{+} - F^{-} \).
Latin indices \( (a, b, c, d, \ldots) \) refer to components in the four-dimensional spacetime, while Greek indices \( (\alpha, \beta, \gamma, \ldots) \) indicate components tangent to the junction hypersurface.
The normal vector to the hypersurface is denoted by \( n^{a} \) and satisfies \( n_{a} n^{a} = \epsilon \), where \( \epsilon = +1 \) for timelike hypersurfaces (spacelike normal) and \( \epsilon = -1 \) for spacelike hypersurfaces (timelike normal).
The induced metric on \( \Sigma \) is represented by \( h_{\alpha\beta} \) and, in general situations, its discontinuity is indicated by \( [h_{\alpha\beta}] \neq 0 \).
These conventions are maintained in all subsequent expressions, ensuring consistency in the geometric formulation and in the definition of generalized derivatives along \( \Sigma \).

The article is organized as follows.  
In Section~\ref{sec2}, we present the fundamental concepts of Colombeau algebras, emphasizing how this structure allows one to rigorously define products and derivatives of distributions.  
In Section~\ref{sec3}, we formulate the globally discontinuous spacetime metric and discuss the geometric implications of the condition \( [g_{ab}] \neq 0 \), including the behavior of the normal vector, the inverse metric, and the induced metric.  
In Section~\ref{sec4}, we explicitly derive the generalized junction conditions from the regularized Einstein equations, identifying the singular contributions and comparing the results with the Darmois–Israel formalism.  
Finally, in Section~\ref{sec5}, we present the conclusions.

\section{Preliminaries}\label{sec2}

To address problems involving products of distributions, such as those that arise when dealing with discontinuous metrics, we employ the framework of Colombeau algebras. This formalism provides a rigorous way to handle expressions that are not well defined in the classical theory of Schwartz distributions.

\subsection{Basic Notions of the Colombeau Formalism}

To rigorously treat products of distributions and discontinuities, we consider families of smooth functions indexed by a parameter \( \varepsilon \in ]0,1] \). Formally, let \( \Omega \subseteq \mathbb{R}^n \) be an open set; we define

\begin{definition}\label{alg-1}
	\(\mathscr{E}(\Omega) \coloneq (\mathscr{C}^\infty(\Omega))^I\) as the set of all functions that, for each \( \varepsilon \in I \), assign a function \( f_\varepsilon \in \mathscr{C}^\infty(\Omega) \), that is, smooth in the variable \( x \in \Omega \).
\end{definition}

From this space of functions, we introduce:

\begin{definition}\label{alg-2}
	\begin{enumerate}
		\item[\(M)\)] \textbf{Moderate functions:}
\begin{eqnarray}
\mathscr{E}_M(\Omega)&\coloneq&\{(u_\varepsilon)_{\varepsilon}\in\mathscr{E}(\Omega)^\infty(\Omega)|\forall~\alpha\in\mathbb{N}_0^n,~\forall~K\subset\subset\Omega
\exists~p\in\mathbb{N},~\exists~\eta>0,~\exists~c>0 \nonumber\\
&&~\mbox{such that}~\|\partial^\alpha u_\varepsilon(x)\|_K \le c \varepsilon^{-p}, \forall 0<\varepsilon<\eta\}\nonumber
\end{eqnarray}

\item[\(N)\)] \textbf{Negligible functions:}
	\begin{eqnarray}
	\mathscr{N}(\Omega)&\coloneq&\{(u_\varepsilon)_{\varepsilon}\in\mathscr{E}_M(\Omega)|\forall~\alpha\in\mathbb{N}_0^n,~\forall~K\subset\subset\Omega\forall~q\in\mathbb{N},~\exists~\eta>0,~\exists~c>0\nonumber\\
&&~\mbox{such that}~ \|\partial^\alpha u_\varepsilon(x)\|_K \leq c\varepsilon^q,~\forall~0<\varepsilon<\eta\}
	\end{eqnarray}
	\end{enumerate}
	
	The simplified Colombeau algebra over \(\Omega\) is then defined as the quotient
	\[
	\mathscr{G}(\Omega) \coloneq \mathscr{E}_M(\Omega) / \mathscr{N}(\Omega).
	\]
\end{definition}

Although this definition of moderate functions, null functions, and of the quotient structure itself may at first glance appear formal and abstract, it reflects a rather concrete idea. Each element of $\mathscr{G}(\Omega)$ essentially represents a class of families of smooth functions $(u_\varepsilon)_\varepsilon$ that provide a regularized description of a potentially singular object, such as a distribution. Moderate families are those whose growth, as $\varepsilon \to 0$, is controlled by a finite power of $\varepsilon$, whereas null families decay faster than any power of $\varepsilon$. Thus, by taking the quotient $\mathscr{E}_M(\Omega)/\mathscr{N}(\Omega)$, one identifies as equivalent all moderate families that differ only by negligible terms, that is, \((u_\varepsilon - v_\varepsilon)_\varepsilon \in \mathscr{N}(\Omega)\).  

Hence, an element of $\mathscr{G}(\Omega)$ is not a single function, but rather an equivalence class of smooth functions that behave in the same way in the limit $\varepsilon \to 0$. This construction makes it possible to handle, within a single framework, both regular functions and singular objects, preserving differential and multiplicative operations that are not well defined in the classical theory of distributions. In language closer to physical terminology, the families $(u_\varepsilon)_\varepsilon$ may be viewed as regularizations of distributions, while the elements of $\mathscr{G}(\Omega)$ represent the generalized limits of these regularizations \cite{Huber2025}.

The Colombeau algebra \(\mathscr{G}(\Omega)\) is therefore commutative, associative, and distributive, contains smooth functions as a faithful subalgebra, incorporates classical distributions as a linear subspace, and satisfies the Leibniz rule for derivatives of products. This structure makes it possible to manipulate products and derivatives of distributions in a consistent way, even when such operations are undefined in the space of classical distributions.

Having defined the Colombeau algebra \(\mathscr{G}(\Omega)\), we can now explicitly present a practical construction of generalized functions. To this end, consider a smooth test function \(\varphi \in C^\infty_c(\Omega) = \mathscr{D}(\Omega)\), with \(\Omega \subset\subset \mathbb{R}^n\), satisfying the fundamental conditions:
\begin{enumerate}[i)]
	\item $\int \varphi(x)\, dx^n = 1$, 	\label{cond-1}
	\item $\int \varphi(x)\, x^\alpha\, dx^n = 0$, for all $|\alpha| \ge 1$. \label{cond-2}
\end{enumerate}

We then define the family of functions
\[
\varphi_\epsilon(x) = \epsilon^{-n} \varphi\left(\frac{x}{\epsilon}\right),
\]
and, from it, the theta net
\[
\theta_\epsilon(x) = (\theta * \varphi_\epsilon)(x) = \int \theta(y) \varphi_\epsilon(x-y) \, dy^n,
\]
which provides a smooth regularization of the Heaviside function \(\theta(x)\). Similarly, the strict delta net is given by
\[
\delta_\epsilon(x) = (\delta * \varphi_\epsilon)(x) = \int \delta(y) \varphi_\epsilon(x-y) \, dy^n~.
\]

Throughout the development of the Colombeau algebra, the need naturally arises to distinguish between different levels of equality between generalized functions. As established by Colombeau~\cite{Colombeau1992}, a fundamental role is played by a relation that may exist between two different elements of \(\mathscr{G}(\Omega)\).

This relation makes it possible to identify in \(\mathscr{G}(\Omega)\) objects that are not nonlinearly identical but are linearly equivalent; this is done by introducing an equivalence relation in \(\mathscr{E}_{M}(\Omega)\) that is weaker than equality in \(\mathscr{G}(\Omega)\) \cite{Grosser2001}.

\begin{definition}\label{ass-1}
An element \(u\) of \(\mathscr{G}(\Omega)\) is said to be associated with \(0\) (denoted \(u \approx 0\)) if
	\begin{equation}
	\lim_{\varepsilon \to 0}\int_{\Omega}u_{\varepsilon}(x)\,\psi(x)\, dx = 0 ~~~~~ \forall \psi\in \mathscr{C}^\infty_{c}(\Omega) = \mathscr{D}(\Omega)~.
	\end{equation}
\end{definition}
The Definition~\ref{ass-1} is independent of the chosen representative \(u_{\varepsilon}\) of \(u\).  
The equivalence relation
\[
u \approx v \;\Leftrightarrow\; u - v \approx 0~,
\]
for \(u, v \in \mathscr{G}(\Omega)\), states that they are associated with each other.  
Explicitly, given representatives \(u_{\varepsilon}, v_{\varepsilon}\) of \(u, v\), we have
\begin{equation}\label{ass-2}
\lim_{\varepsilon \to 0}\int_{\Omega} \big(u_{\varepsilon}(x) - v_{\varepsilon}(x)\big)\,\psi(x)\, dx = 0~.
\end{equation}
If this property holds for one particular pair of representatives, then it also holds for any other respective pair.

Both members of an association may be freely differentiated:
\begin{equation}\label{ass-3}
u \approx v \;\Rightarrow\; Du \approx Dv~,
\end{equation}
for any partial differential operator \(D\).

Having established the notion of association in \(\mathscr{G}(\mathbb{R})\), we can derive the relation \(\theta^n \approx \theta\). Thus,
\begin{equation}
\theta^n \approx \theta 
\quad \Rightarrow \quad 
\theta^{\,n-1}\theta' \approx \frac{1}{n}\theta',
\end{equation}
and for \(n=1\) we obtain
\begin{equation}\label{H-delta}
\theta \delta \approx \frac{1}{2}\delta,
\end{equation}
where \(\theta' = \delta\).  
This association provides a ``shadow'' of the corresponding operation in the classical distributional setting, allowing one to interpret the product \(\theta\delta\), which does not exist in \(\mathscr{D}'\).

From \eqref{H-delta}, we can compute \(\theta\delta'\) using the Leibniz rule:
\begin{equation}
(\theta\delta)' 
= \theta'\delta + \theta\delta' 
\approx \frac{1}{2}\delta'
\quad \Rightarrow \quad 
\delta\delta + \theta\delta' \approx \frac{1}{2}\delta',
\end{equation}
which yields
\begin{equation}\label{H-delta-prime}
\theta\delta' \approx \frac{1}{2}\delta' - \delta^{2}.
\end{equation}
This expression lies entirely within the framework of \(\mathscr{G}(\mathbb{R})\), and it is well defined even though neither \(\theta\delta'\) nor \(\delta^2\) exists in the classical theory of distributions.

\subsubsection*{The square of the delta distribution}

Within the Colombeau algebra, \(\delta^{2}\) is defined as the product \(\delta \cdot \delta\), whose representative is given by
\begin{equation}
\varphi_{\varepsilon}^{2}(x)
= \left( \frac{1}{\varepsilon}\,\varphi\!\left(\frac{x}{\varepsilon}\right) \right)^{2}
= \frac{1}{\varepsilon^{2}}\,\varphi^{2}\!\left(\frac{x}{\varepsilon}\right),
\end{equation}
where \(\varphi \in \mathscr{D}(\mathbb{R})\) satisfies conditions~\eqref{cond-1} and~\eqref{cond-2}.  
The action of \(\delta^{2}\) on a test function \(\psi\) is then
\begin{equation}
\lim_{\varepsilon\to 0}
\int \varphi_{\varepsilon}^{2}(x)\,\psi(x)\,dx
= \lim_{\varepsilon\to 0}
\frac{1}{\varepsilon}
\int \varphi^{2}(u)\,\psi(\varepsilon u)\,du
\;\rightarrow\; \infty,
\end{equation}
whenever \(\psi(0)\neq 0\).  
Therefore, \(\delta^{2}\) is not associated with any classical distribution, but it is well defined as an element of \(\mathscr{G}(\mathbb{R})\).  

Thus, \(\delta^{2}\in \mathscr{G}(\mathbb{R})\) with representative \(\varphi_{\varepsilon}^{2}\).

Let us now examine the properties of \(\delta^{2}\) in parallel with those of the Dirac delta.  
In this case we have
\begin{equation}\label{dirac-quad-prop}
\varphi_{\varepsilon}^{2}(x)=0 \quad \text{for}~ x\neq 0,
\qquad
\int_{-\infty}^{\infty}\varphi_{\varepsilon}^{2}(x)\,dx \;\rightarrow\; \infty
\quad \text{as}~ \varepsilon\to 0.
\end{equation}
More explicitly,
\[
\int \delta^{2}(x)\,dx \sim \frac{1}{\varepsilon},
\]
which shows that this integral diverges in the classical limit, yet it possesses a well–defined meaning as a generalized number in the sense of Colombeau~\cite{Colombeau1992}.

Therefore, the integral of $\delta^2$ is not a real number, but rather a generalized number\footnote{In Colombeau algebras, generalized numbers are classes of moderate sequences of real (or complex) numbers, whose behavior is controlled in terms of powers of $\varepsilon$. These numbers extend the classical reals and allow one to represent quantities such as (\ref{dirac-quad-prop}), which diverge in the traditional context but remain well defined within the structure $\overline{\mathbb{R}}$. The interested reader may find further details on the subject in \cite{Aragona2001, Ronaldo2006}.}, belonging to the algebra $\overline{\mathbb{R}}$, and its behavior is well represented by $\varphi_\varepsilon^2$, whose integral grows inversely proportional to $\varepsilon$.

Now, let us present the product \(H\delta^2.\) Note that \(\delta^2 = \theta'\theta'.\) Since $\mathscr{G}$ is a commutative and associative algebra, we can perform the following algebraic manipulations without difficulty
\begin{eqnarray}
\theta\delta^2 = \theta(\theta'\theta') = (\theta\theta')\theta' &=& (\theta\delta)\delta\nonumber\\
&\approx& \left(\dfrac{1}{2}\delta\right)\delta \nonumber\\
&\approx& \dfrac{1}{2}\delta^2~.
\end{eqnarray}
Therefore, we conclude that
\begin{equation}\label{assoc5}
\theta\delta^2 \approx \dfrac{1}{2}\delta^2.
\end{equation}

It should be emphasized that, as far as we know, explicit approaches of this kind, which deal with products involving powers of the delta distribution, such as $\delta^2$ or $\theta\delta^2$, in a constructive manner within the very definition of association, rarely appear in the specialized literature. Even in the fundamental works of Colombeau~\cite{Colombeau1992} and in later developments (see \cite{Ober1992, Grosser2001}), the emphasis falls mainly on the structural analysis of the space $\mathscr{G}$ and on demonstrating that $\delta^2$ does not admit association with classical distributions, without the explicit execution of such products. Subsequent works, such as that of Miteva et~al.~(2016) \cite{Miteva2016}, devoted to the study of products in Colombeau algebras, focus primarily on combinations involving derivatives of the delta distribution and smooth factors, restricting themselves to cases in which the resulting product admits an associated distribution. In other words, the concept of Colombeau product adopted by these authors applies only when the operation between two distributions embedded in $\mathscr{G}$ preserves a counterpart in the space $\mathscr{D}'$. In contrast, in the present work we consider purely generalized products, whose existence is guaranteed solely within the algebra $\mathscr{G}$, as in the expressions $\theta\delta'$ and $\theta\delta^2$, which have no counterparts within the classical distributional framework.

Although the products $\theta\delta'$ and $\theta\delta^2$ can be obtained algebraically from the previously established association relations, a direct proof of expressions (\ref{H-delta-prime}) and (\ref{assoc5}) based on Definition~\ref{ass-1}, by means of Eq.~(\ref{ass-2}), is presented below.

\begin{lemma}\label{lema1}
	The function \(f(u) = \int_{-\infty}^{u}\varphi(s)\, ds - \dfrac{1}{2},\) where \(\varphi_\varepsilon(x) = \dfrac{1}{\varepsilon}\varphi\!\left(\dfrac{x}{\varepsilon}\right)\) is a delta net, is an odd function.
\end{lemma}

\begin{proof}
To show this, we need to prove that \(f(-u) = -f(u)\). Indeed,
	\begin{eqnarray}
	f(-u) &=& \int_{-\infty}^{-u}\varphi(s)ds - \dfrac{1}{2} - \dfrac{1}{2} + \dfrac{1}{2} \nonumber\\
	&=& \int_{-\infty}^{-u}\varphi(s)ds - 1 + \dfrac{1}{2} \nonumber\\
	&=& \int_{-\infty}^{-u}\varphi(s)ds - \int_{-\infty}^{+\infty}\varphi(s)ds + \dfrac{1}{2}\nonumber\\
	&=& \int_{-\infty}^{-u}\varphi(s)ds - \left(\int_{-\infty}^{-u}\varphi(s)ds+\int_{-u}^{+\infty}\varphi(s)ds\right) + \dfrac{1}{2}\nonumber\\
	&=&  -\int_{-u}^{+\infty}\varphi(s)ds + \dfrac{1}{2}~.\nonumber
	\end{eqnarray}
Making the following substitutions
	\begin{equation}
	\begin{cases}
	t = -s \Rightarrow dt = -ds, \\
	s = -u \Rightarrow t = u, \\
	s = +\infty \Rightarrow t = - \infty
	\end{cases}
	\end{equation}
we obtain
	\begin{eqnarray}
	f(-u) &=& - \int_{u}^{-\infty}\varphi(-s)(-ds) + \dfrac{1}{2} \nonumber\\
	&=& \int_{u}^{-\infty}\varphi(s)ds + \dfrac{1}{2}\nonumber\\
	&=& -\int_{-\infty}^{u}\varphi(s)ds + \dfrac{1}{2}\nonumber\\
	&=& - \left(\int_{-\infty}^{u}\varphi(s)ds - \dfrac{1}{2}\right)\nonumber\\
	&=& - f(u)~.
	\end{eqnarray}
Therefore, \(f(-u) = -f(u)\), which proves the statement.
\end{proof}
We will use this Lemma~\ref{lema1} to prove the following theorem.

\begin{theorem}\label{produto-Hdelta^2}
	Let \(H\) be a Heaviside-type distribution and \(\delta\) a Dirac delta-type distribution. Then
	\begin{itemize}
		\item[(i)] $H\delta^2\approx \dfrac{1}{2}\delta^2$
		\item[(ii)] $H\delta'\approx \dfrac{1}{2}\delta' - \delta^2$
	\end{itemize}
\end{theorem}

\begin{proof}
	\textit{(i)} To show that \(H\delta^2 \approx \dfrac{1}{2}\delta^2,\) we need to prove that
	\[H\delta^2 - \dfrac{1}{2}\delta^2 \approx 0~,\]
or equivalently
	\begin{equation}\label{assoc1}
	\left(H - \dfrac{1}{2}\right)\delta^2 \approx 0.
	\end{equation}
We know that
	\[H_{\varepsilon}(x) = \int_{-\infty}^{x}\dfrac{1}{\varepsilon}\varphi\left(\dfrac{t}{\varepsilon}\right)dt = \int_{-\infty}^{x/\varepsilon}\varphi(s)ds\]
and
	\[\delta_{\varepsilon}(x) = \dfrac{1}{\varepsilon}\varphi\left(\dfrac{x}{\varepsilon}\right)\]
are representatives of \(H\) and \(\delta\), respectively, in \(\mathscr{G}(\mathbb{R})\). Therefore,
	\[g_{\varepsilon}(x) = H_{\varepsilon}(x) - \dfrac{1}{2} = \int_{-\infty}^{x/\varepsilon}\varphi(s)ds - \dfrac{1}{2}\]
and
	\[\delta_{\varepsilon}^2(x) = \dfrac{1}{\varepsilon^2}\varphi^2\left(\dfrac{x}{\varepsilon}\right)\]
are representatives of \(H - 1/2\) and \(\delta^2\), respectively. Thus, to show (\ref{assoc1}), we must show that
	\begin{equation}\label{assoc2}
	\lim_{\varepsilon\rightarrow 0}\langle \left(H_{\varepsilon} - \dfrac{1}{2}\right)\delta^2, \psi\rangle = 0, \forall \psi\in \mathscr{C}^\infty_{c}(\mathbb{R}).
	\end{equation}
Now,
	\begin{eqnarray}\label{assoc3}
	\lim_{\varepsilon\rightarrow 0}\langle \left(H_{\varepsilon} - \dfrac{1}{2}\right)\delta^2, \psi\rangle &=& \lim_{\varepsilon\rightarrow 0}\int_{\mathbb{R}}\left(H_{\varepsilon} - \dfrac{1}{2}\right)\delta^2_{\varepsilon}(x)\psi(x)dx \nonumber\\
	&=& \lim_{\varepsilon\rightarrow 0}\int_{\mathbb{R}}g_{\varepsilon}(x)\delta^2_{\varepsilon}(x)\psi(x)dx \nonumber\\
	&=& \lim_{\varepsilon\rightarrow 0}\int_{\mathbb{R}}g_{\varepsilon}(x)\dfrac{1}{\varepsilon^2}\varphi^2\left(\dfrac{x}{\varepsilon}\right)\psi(x)dx
	\end{eqnarray}
Assuming \(\mathrm{supp}(\varphi)\subseteq [-l,l]\) and making the substitution
	\[-l\leq u = \dfrac{x}{\varepsilon}\leq l \Rightarrow x = \varepsilon u \Rightarrow dx = \varepsilon du\]
which substituted into (\ref{assoc3}) yields
	\begin{eqnarray}
	\lim_{\varepsilon\rightarrow 0}\langle \left(H_{\varepsilon} - \dfrac{1}{2}\right)\delta^2, \psi\rangle &=& \lim_{\varepsilon\rightarrow 0}\int_{\mathbb{R}}g_{\varepsilon}(u)\dfrac{1}{\varepsilon^2}\varphi^2(u)\psi(\varepsilon u)\varepsilon du \nonumber\\
	&=& \lim_{\varepsilon\rightarrow 0}\dfrac{1}{\varepsilon}\int_{\mathbb{R}}g_{\varepsilon}(u)\varphi^2(u)\psi(\varepsilon u) du
	\end{eqnarray}
by Lebesgue's dominated convergence theorem,
	\begin{eqnarray}
	\lim_{\varepsilon\rightarrow 0}\langle \left(H_{\varepsilon} - \dfrac{1}{2}\right)\delta^2_{\varepsilon}, \psi\rangle &=& \psi(0)\int_{\mathbb{R}}\left(\int_{-\infty}^{u}\varphi(s)ds - \dfrac{1}{2}\right)\varphi^2(u) du \nonumber\\
	&=& 0~,
	\end{eqnarray}
since \[g_{\varepsilon}(x) = \int_{-\infty}^{u}\varphi(s)ds - \dfrac{1}{2}\] is an odd function and $\varphi^2(u)$ is an even function, therefore, $g_{\varepsilon}(u)\varphi^2(u)$is an odd function and
	\[\int_{-l}^{l}\left(\int_{-\infty}^{u}\varphi(s)ds - \dfrac{1}{2}\right)\varphi^2(u) du = 0.\] And this shows the equality in (\ref{assoc3}). Therefore, \textit{(i)} holds.
	
\textit{(ii)} To show that \(H\delta' \approx \dfrac{1}{2}\delta' - \delta^2,\) we need to show that
	\begin{equation}\label{assoc4}
	H\delta' - \dfrac{1}{2}\delta' + \delta^2 \approx 0,
	\end{equation}
if and only if
	\begin{equation}\label{assoc4'}
	\left(H - \dfrac{1}{2}\right)\delta' + \delta^2 \approx 0.
	\end{equation}
A representative of \(\delta'\) in \(\mathscr{G}(\mathbb{R})\) is \[\delta'_{\varepsilon}(x) = \dfrac{1}{\varepsilon^2}\varphi'\left(\dfrac{x}{\varepsilon}\right).\] To show (\ref{assoc4'}), we must show that
	\begin{equation}\label{assoc5a}
	\lim_{\varepsilon \to 0} \langle \left(H_{\varepsilon} - \dfrac{1}{2}\right)\delta'_{\varepsilon} + \delta_{\varepsilon}^2, \psi\rangle = 0, \forall \psi \in \mathscr{C}^\infty_{c}(\mathbb{R}).
	\end{equation}
Now,
	\begin{eqnarray}\label{assoc6}
	\lim_{\varepsilon \to 0} \langle \left(H_{\varepsilon} - \dfrac{1}{2}\right)\delta'_{\varepsilon} + \delta_{\varepsilon}^2, \psi\rangle &=& \lim_{\varepsilon \to 0} \langle g_{\varepsilon}\delta'_{\varepsilon} + \delta^2_{\varepsilon}, \psi\rangle \nonumber\\
	&=& \lim_{\varepsilon \to 0} \int_{\mathbb{R}} (g_{\varepsilon}(x)\delta'_{\varepsilon}(x) + \delta^2_{\varepsilon}(x))\psi(x) dx \nonumber\\
	&=& \lim_{\varepsilon \to 0} \left(\int_{\mathbb{R}} (g_{\varepsilon}(x)\psi(x))\delta'_{\varepsilon}(x) dx + \int_{\mathbb{R}}\delta^2_{\varepsilon}(x)\psi(x) dx \right)
	\end{eqnarray}
Let's solve the first integral
	\begin{eqnarray}\label{assoc7}
	\int_{\mathbb{R}} (g_{\varepsilon}(x)\psi(x))\delta'_{\varepsilon}(x) dx &=& -\int_{\mathbb{R}}(g_{\varepsilon}(x)\psi(x))'\delta_{\varepsilon}(x) dx \nonumber\\
	&=& - \int_{\mathbb{R}}(g'_{\varepsilon}(x)\psi(x) + g_{\varepsilon}(x)\psi'(x))\delta_{\varepsilon}(x) dx \nonumber\\
	&=& - \int_{\mathbb{R}} (\delta_{\varepsilon}(x)\psi(x) + g_{\varepsilon}(x)\psi'(x))\delta_{\varepsilon}(x) dx \nonumber\\
	&=& - \int_{\mathbb{R}} \delta^2_{\varepsilon}(x)\psi(x) dx -  \int_{\mathbb{R}}g_{\varepsilon}(x)\psi'(x)\delta_{\varepsilon}(x) dx
	\end{eqnarray}
When we substitute (\ref{assoc7}) into (\ref{assoc6}), the last term of (\ref{assoc6}) cancels out with the first term of (\ref{assoc7}), obtaining	
	\begin{eqnarray}\label{assoc8}
	\lim_{\varepsilon\rightarrow 0}\langle \left(H_{\varepsilon} - \dfrac{1}{2}\right)\delta'_{\varepsilon} + \delta^2, \psi\rangle &=& \lim_{\varepsilon\rightarrow 0} \left(- \int_{\mathbb{R}}g_{\varepsilon}(x)\psi'(x)\delta_{\varepsilon}(x) dx\right) \nonumber\\
	&=& - \lim_{\varepsilon\rightarrow 0}\int_{\mathbb{R}}g_{\varepsilon}(x)\psi'(x)\delta_{\varepsilon}(x) dx \nonumber\\
	&=& - \lim_{\varepsilon\rightarrow 0} \int_{\mathbb{R}}\left(\int_{-\infty}^{x/\varepsilon}\varphi(s) ds - \dfrac{1}{2}\right)\dfrac{1}{\varepsilon}\varphi\left(\dfrac{x}{\varepsilon}\right)\psi'(x) dx \nonumber\\
	&=& - \lim_{\varepsilon\rightarrow 0}\dfrac{1}{\varepsilon} \int_{\mathbb{R}}\left(\int_{-\infty}^{x/\varepsilon}\varphi(s) ds - \dfrac{1}{2}\right)\varphi\left(\dfrac{x}{\varepsilon}\right)\psi'(x) dx
	\end{eqnarray}
Assuming that \(\text{supp}(\varphi) \subseteq [-l, l]\) and letting $-l \leq u = x/\varepsilon \leq l$, we have $u = \varepsilon u\Rightarrow dx = \varepsilon du$. Now, substituting in (\ref{assoc8}), we get	
\begin{eqnarray}
	\lim_{\varepsilon\rightarrow 0}\langle \left(H_{\varepsilon} - \dfrac{1}{2}\right)\delta'_{\varepsilon} + \delta^2, \psi\rangle &=& - \lim_{\varepsilon\rightarrow 0} \dfrac{1}{\varepsilon}\int_{-l}^{l}\left(\int_{-\infty}^{u}\varphi(s)ds - \dfrac{1}{2}\right)\varphi(u)\psi'(\varepsilon u) \varepsilon du \nonumber \\
	&=& - \psi'(0)\int_{-l}^{l} \left(\int_{-\infty}^{u}\varphi(s)ds - \dfrac{1}{2}\right) \varphi(u) du \nonumber\\
	&=& 0~,
\end{eqnarray}
since \[\int_{-l}^{l} \left(\int_{-\infty}^{u}\varphi(s)ds - \dfrac{1}{2}\right) \varphi(u) du = 0.\]
\end{proof}

As previously mentioned in the introduction, Racskó~\cite{Racsko2024} emphasizes the inherent difficulties in multiplying singular distributions by discontinuous functions, a situation that frequently arises in junction problems. Within the classical theory of distributions, such operations lack mathematical validity, leading many authors to adopt ad hoc prescriptions, such as interpreting the value of a discontinuous function at the junction as the average of its one-sided limits. Racskó explicitly criticizes this approach, showing that it may lead to incorrect results, and therefore develops a framework that avoids any “illegitimate” product or regularization involving delta functions. However, his structure remains restricted to situations in which stronger singularities, such as $\delta^{2}$, do not explicitly appear.

In contrast, within Colombeau's algebra, these operations are rigorously defined. Products involving $\delta^2$, $\theta\delta'$ or $\theta\delta^2$ are treated as well-defined generalized elements, allowing such ultrasingular terms to be manipulated without ad~hoc assumptions or ambiguities. In this sense, the algebraic structure of $\mathscr{G}$ provides a mathematically safe environment for exploring configurations that lie beyond the scope of classical distribution theory, ensuring consistency even in regimes where discontinuities in the metric or in its derivatives become physically relevant.

This mathematical robustness becomes particularly important when one seeks to understand the physical role of ultrasingular terms. In this context, it is worth noting that the presence of powers of distributions, such as $\delta^2$, has also been investigated in situations of physical interest. Oberguggenberger~\cite{Oberguggenberger2001} considered the term $\iota(\delta)^2$ as initial data in the one-dimensional wave equation, showing that although the corresponding solution is not $\mathscr{C}^\infty$-regular inside the light cone, it remains $\mathscr{G}^\infty$-regular in that domain. This result shows that even powers of distributions can be consistently handled within Colombeau's algebra, preserving the mathematical coherence of singular physical systems.

This analysis reinforces the idea that ultrasingular terms, although devoid of classical meaning, can represent legitimate physical idealizations when treated in a generalized function framework. In the present work, we explore this perspective in the geometric setting, where the presence of $\delta^2$ arises naturally in junctions involving a discontinuous metric.

\section{Formulation of the discontinuous spacetime metric}\label{sec3}

To describe the coupling between two regions of space-time with distinct geometries, we consider two pseudo-Riemannian manifolds $(M^+, g^+_{ab})$ and $(M^-, g^-_{ab})$, with metrics $g^+_{ab}$ and $g^-_{ab}$ defined respectively on $M^+$ and $M^-$, which are submanifolds of a larger space-time $M$, such that $M^+ \cup M^- \subseteq M$. These regions are separated by a common hypersurface $\Sigma$, which can be regarded as the junction interface between the two geometries.

Unlike the traditional case, in which the metric is assumed to be continuous on the hypersurface ($[g_{ab}] = 0$), here we work with a genuine geometric break, that is, with a globally discontinuous metric of the form:
\begin{equation}
g_{ab}(x) = \theta(l)\, g^{+}_{ab}(x) + \theta(-l)\, g^{-}_{ab}(x),
\end{equation}
where $l(x)$ is a smooth function that vanishes on $\Sigma$, and $\theta(l)$ is the Heaviside function. This definition implies that, although the metric is well defined on each side of $\Sigma$, its restriction to the hypersurface exhibits a discontinuity
\begin{equation}\label{condiction1}
[g_{ab}] \neq 0~.
\end{equation}
This type of construction requires additional care, since the derivation of geometric quantities such as the Christoffel symbols, curvature tensors, and the energy-momentum tensor itself involves the presence of terms like $\delta(l)$, $\delta'(l)$ and $\delta^2(l)$, the latter of which has no well-defined meaning in the classical theory of distributions. For this reason, the treatment presented here will be carried out within the framework of Colombeau algebras, which provide a mathematically rigorous environment for handling such expressions.

In order for the structure above to be treated in a mathematically well-defined way, we introduce a smooth regularization of the Heaviside function, denoted by $\theta_{\varepsilon}(l)$. The ambient metric is then defined by
\begin{equation}\label{metricaregularizada}
g_{\varepsilon ab}(x) = \theta_{\varepsilon}(l) g^{+}_{ab}(x) + \theta_{\varepsilon}(-l) g^{-}_{ab}(x)~,
\end{equation}
where $l = l(x)$ is a smooth function such that $l(x) = 0$ defines the hypersurface $\Sigma$. Since $\theta_{\varepsilon}(-l) = 1 - \theta_{\varepsilon}(l)$, this regularized metric can be interpreted as a smooth interpolation between $g^{+}_{ab}$ and $g^{-}_{ab}$ along the normal to $\Sigma$.

This metric is then defined in such a way that
\begin{eqnarray}
\langle g_{\varepsilon ab}, \varphi \rangle &=& \int_{M} g_{\varepsilon ab}(x) \varphi(x) dx \nonumber \\
&=& \int_{M} (\theta_{\varepsilon}(l) g^{+}_{ab}(x) + \theta_{\varepsilon}(-l) g^{-}_{ab}(x)) \varphi(x) \nonumber \\
&=& \int_{M} \theta_{\varepsilon}(l) g^{+}_{ab}(x)dx + \int_{M}\theta_{\varepsilon}(-l) g^{-}_{ab}(x)) \varphi(x)dx\nonumber \\
&=& \langle \theta_{\varepsilon}(l) g^{+}_{ab}(x), \varphi \rangle +  \langle \theta_{\varepsilon}(-l) g^{-}_{ab}(x), \varphi \rangle~,
\end{eqnarray}
in the limit $\varepsilon \to 0$, $\theta_\varepsilon \to \theta$, and the regularized metric tends, in a generalized sense, to the discontinuous metric
\begin{equation}
\lim_{\varepsilon \to 0} \langle g_{\varepsilon ab}, \varphi \rangle \equiv \langle g_{ab}, \varphi \rangle \equiv \langle \theta(l) g^{+}_{ab}, \varphi \rangle +  \langle \theta(-l) g^{-}_{ab}, \varphi \rangle~
\end{equation}
and so
\begin{equation}
g_{\varepsilon ab}(x) \approx \theta(l) g^+{ab}(x) + \theta(-l) g^-{ab}(x)~,
\end{equation}
where the symbol $\approx$ denotes equality in the weak sense in Colombeau algebras (see Section \ref{sec2}).

This definition of the ambient metric as an element of a Colombeau algebra will allow us, in the following sections, to rigorously calculate geometric quantities even in the presence of explicit discontinuities in the metric. Furthermore, this formulation avoids ambiguities associated with products of distributions already mentioned, which naturally arise when deriving the metric in the neighborhood of $\Sigma$.

\subsection{Geometric implications of metric discontinuity}

Let's evaluate the impact of considering \( [g_{ab}] \neq 0 \). Furthermore, we will present how we will treat the notion of normal and tangent vectors in this context.

Allowing \( [g_{ab}] \neq 0 \) means that
\begin{enumerate}
	\item The metric \( g_{ab} \) is discontinuous on the hypersurface \( \Sigma \), that is, \( g^+_{ab} \neq g^-_{ab} \).
	
	\item  The induced metric \( h_{\alpha\beta} = g_{ab} e^a_\alpha e^b_\beta \) will also be discontinuous
	\begin{equation}
	[h_{\alpha\beta}] = [g_{ab}] e^a_\alpha e^b_\beta \neq 0~,
	\end{equation}
but this raises questions about the behavior of \( e^a_\alpha\) and will be discussed later.
	
	\item The inverse metric $g^{ab}$ is also discontinuous and can be expressed distributively as
	\begin{equation}
	g^{ab}(x) = \theta(l) g^{ab+}(x) + \theta(-l) g^{ab-}(x).
	\end{equation}
	Its regularized form is
	\begin{equation}\label{metric-invers-gen}
	g^{ab}_\varepsilon = \theta_\varepsilon(l) g^{ab+} + \theta_\varepsilon(-l) g^{ab-}~.
	\end{equation}
However, it is important to highlight that, due to the non-linearity of the matrix inversion operation, the function $g^{ab}_\varepsilon$ is not, in general, the exact inverse of $g_{\varepsilon ab}$, that is, $(g_{\varepsilon ab})^{-1} \neq g^{ab}_\varepsilon.$
	
Even so, in Colombeau's formalism, the identity $g^{ab}_\varepsilon g_{\varepsilon bc} = \delta^a_c$ is recovered in the weak sense, through association.
	
To verify this property, we observe that
	\begin{eqnarray}\label{prop.metric}
	g^{ab}_\varepsilon g_{\varepsilon bc} &=& \left(\theta_\varepsilon(l) g^{ab+} + \theta_\varepsilon(-l) g^{ab-}\right) \left( \theta_\varepsilon(l) g^{+}_{bc} + \theta_\varepsilon(-l) g^{-}_{bc}\right) \nonumber\\
	&=& \theta^2_\varepsilon(l) g^{ab+}g^{+}_{bc} + \theta_\varepsilon(l)\theta_\varepsilon(-l)g^{ab+}g^{-}_{bc} + \theta_\varepsilon(-l)\theta_\varepsilon(l)g^{ab-}g^{+}_{bc} +  \theta^2_\varepsilon(-l) g^{ab-}g^{-}_{bc} \nonumber\\
	&\approx& \theta_\varepsilon(l)g^{ab+}g^{+}_{bc} +  \theta_\varepsilon(-l) g^{ab-}g^{-}_{bc} \nonumber\\
	&\approx& \theta_\varepsilon(l)\delta^{a}_{c} +  \theta_\varepsilon(-l) \delta^{a}_{c} \nonumber\\
	&\approx& \delta^{a}_{c}~.
	\end{eqnarray}
Here we use the properties \[	\theta_\varepsilon^2(l) \approx \theta_\varepsilon(l), \quad \theta_\varepsilon(l)\theta_\varepsilon(-l) \approx 0.\]

The Eq. (\ref{prop.metric}) shows that, in the limit $\varepsilon \rightarrow 0$, the regularized metric satisfies
	\begin{itemize}
		\item For $l > 0$: $g_{\varepsilon bc} \rightarrow g^+_{bc}, \quad g^{ab}_\varepsilon \rightarrow g^{ab+}$ and $g^{ab+} g^+_{bc} = \delta^a_c$;
		\item For $l < 0$: $g_{\varepsilon bc} \rightarrow g^-_{bc}, \quad g^{ab}_\varepsilon \rightarrow g^{ab-}$ and $g^{ab-} g^-_{bc} = \delta^a_c$.
	\end{itemize}
	
Therefore, the identity relation $g^{ab}_\varepsilon g_{\varepsilon bc} \approx \delta^a_c$ remains valid in the weak sense, that is, in the context of Colombeau generalized functions.
\end{enumerate}

In the traditional approach to Darmois-Israel junction conditions, the normal vector is defined as
\[n_i = \epsilon \partial_i l, \quad n^i n_i = \epsilon = \mp 1,\] and is continuous across the hypersurface (\( [n^i] = 0 \)), because \( l \) is a continuous coordinate, like the proper distance along geodesics orthogonal to \( \Sigma \), for example.

By assuming \( [g_{ab}] \neq 0 \) the continuity of \( n_i \) can still be maintained, since \( n_i = \epsilon \partial_i l \) depends only on \( l \), and not directly on the metric. Thus, the geometric definition of \( n_i \) as the gradient of \( l \) is not affected by the discontinuity of the metric.

Although the covariant normal vector \(n_i = \epsilon\, \partial_i l\) remains continuous even when \([g_{ab}] \neq 0\), its contravariant form \(n^i = g^{ij} n_j\) is, in general, discontinuous, since it inherits the jump of the inverse metric. This is a natural geometric consequence of the setting rather than an ambiguity: it simply reflects the fact that any object depending on \(g^{ab}\) must exhibit the same degree of regularity (or irregularity) as the metric itself.

However, the normalization \( n^i n_i = \epsilon \) involves the discontinuous inverse metric \( g^{ij} \). This means that
\begin{equation}
n^i n_i = g^{ij} n_i n_j = \theta(l) g^{ij+} n_i n_j + \theta(-l) g^{ij-} n_i n_j~.
\end{equation}
So, just like in the standard case, we need that
\begin{equation}\label{cond.inv1}
g^{ij+} n_i n_j = g^{ij-} n_i n_j = \epsilon~.
\end{equation}
Since \( n_i \) is the same on both sides, that is, continuous, this imposes a restriction on the discontinuity of the inverse metric \( [g^{ij}] \). From Eq. (\ref{cond.inv1}), we have
\begin{eqnarray}
(g^{ij+} - g^{ij-}) n_i n_j &=& g^{ij+} n_i n_j - g^{ij-} n_i n_j \nonumber \\
&=& \epsilon - \epsilon = 0~.
\end{eqnarray}
Therefore,
\begin{equation}\label{cond.inv2}
[g^{ij}] n_i n_j = 0~.
\end{equation}
This condition is the minimum requirement for the normalization $n^i n_i = \epsilon$ to remain valid, even with $g^{ij}$ discontinuous.

It is important to emphasize that the condition established in Eq.~(\ref{cond.inv2}) is a restriction on the discontinuity \( [g^{ij}] \), but does not imply that \( [g^{ij}] \) is zero in all directions. It expresses that the projection of \( [g^{ij}] \) in the direction of \( n_i \) must be zero. In geometric terms, \( [g^{ij}] n_i n_j \) is a scalar that results from the contraction of \( [g^{ij}] \) with \( n_i n_j \).

For \([g^{ij}] n_i n_j = 0\), \([g^{ij}]\) must be orthogonal to \(n_i n_j\) in the tensor sense. This does not mean that \([g^{ij}] = 0\), but rather that \([g^{ij}]\) cannot have components in the direction of \(n_i n_j\). It is merely a restriction on the discontinuity of the inverse metric, ensuring that the normalization of \(n_i\) is consistent on both sides of \( \Sigma \), it only affects a specific combination of the components of \([g^{ij}]\).

In the traditional approach, the tangential vectors \( e^a_\alpha = \partial x^a / \partial y^\alpha \) are defined as the derivatives of the spacetime coordinates \( x^a \) with respect to the intrinsic coordinates of the hypersurface \( y^\alpha \). Furthermore, it is assumed that \( [e^a_\alpha] = 0 \), that is, the same coordinate system is adopted on both sides of \( \Sigma \).

When we consider \( [g_{ab}] \neq 0 \), then \( [h_{\alpha\beta}] = [g_{ab}] e^a_\alpha e^b_\beta \neq 0 \), which means that the induced metric is discontinuous. This raises the question: should the vectors \( e^a_\alpha \) be continuous or not? Geometrically, \( e^a_\alpha \) defines the tangent basis to the hypersurface. If we maintain the same choice of coordinates \( y^\alpha \) on both sides, then \( [e^a_\alpha] = 0 \), and the discontinuity in \( h_{\alpha\beta} \) comes entirely from \( [g_{ab}] \). However, if we allow the coordinates \( y^\alpha \) to be different on each side (that is, we use different coordinate systems for \( \Sigma \) in each region), then \( e^a_\alpha \) can be discontinuous: \( e^{a+}_\alpha \neq e^{a-}_\alpha \). This would bring complexity to the analysis, as we would need to define a relationship between the two coordinate systems.

Therefore, we maintain \( [e^a_\alpha] = 0 \), that is, we use the same coordinate system \( y^\alpha \) on both sides of \( \Sigma \). This maintains the consistency of the standard case and avoids additional complications in the calculations.

\section{Deduction of junction conditions for discontinuous metrics}\label{sec4}

In this section, we apply the regularized discontinuous metric to derive the generalized junction conditions in Einstein's equations. To do this, we analyze how geometric entities such as Christoffel symbols, Riemann and Ricci tensors, and Ricci scalars behave in the presence of genuine discontinuities. The goal is to identify the singular contributions to the energy-momentum tensor associated with these discontinuities.

Analogously to the standard case, when we differentiate Eq.~(\ref{metricaregularizada}) we obtain a proportional contribution to the function $\delta_{\varepsilon}(l)$.
\begin{equation}\label{deriv-metric2}
\partial_{c}g_{\varepsilon ab} = \theta_{\varepsilon}(l)\partial_{c}g^{+}_{ab}+\theta_{\varepsilon}(-l)\partial_{c}g^{-}_{ab} + \epsilon\delta_{\varepsilon}(l)[g_{ab}]n_{c}~.
\end{equation}

In the context of Colombeau's generalized function algebras, there is no ambiguity in the products of distributions that have arisen from this, so we will not require the discontinuity in the metric in $\Sigma$ to disappear, in order to have the most general case. Thus we have $[g_{ab}] \neq 0.$ Since $[g_{ab}] = g^{+}_{ab} - g^{-}_{ab}$, in this case we have $g^{+}_{ab} \neq g^{-}_{ab}$. Physically, this means that the metric is genuinely discontinuous\footnote{The term ``genuine discontinuity'' is used to distinguish the actual discontinuity of the metric represented by an effective jump between $g^{+}_{ab}$ and $g^{-}_{ab}$ on the $\Sigma$ hypersurface from mere apparent discontinuities, which may arise from a poor choice of coordinates or ill-defined transformations. In the context of the standard Darmois-Israel formalism, many apparent discontinuities can be removed by requiring continuity of $g_{ab}$ and good adaptability of the coordinate system. Here, by explicitly relaxing $[g_{ab}] \neq 0$, we are dealing with an intrinsic geometric discontinuity, that is, a physical break and not just a coordinate break.}
and it does not require that the induced metric $h_{\alpha\beta}$ be the same on both sides of $\Sigma$.

From Eq.(\ref{deriv-metric2}) the generalized Christoffel symbols can be determined as
\begin{equation}\label{christoffel-generalized}
\Gamma^{a}_{\varepsilon bc} \approx \theta_{\varepsilon}(l)\Gamma^{a+}_{bc} + \theta_{\varepsilon}(-l)\Gamma^{a-}_{bc} + B^{a}_{bc}\delta_{\varepsilon}(l)~,
\end{equation}
where
\begin{equation}
\Gamma^{a+}_{bc} = \dfrac{1}{2}g^{ad+}(\partial_{c}g^{+}_{bd} + \partial_{b}g^{+}_{cd} - \partial_{d}g^{+}_{bc})~,
\end{equation}
\begin{equation}
\Gamma^{a-}_{bc} = \dfrac{1}{2}g^{ad-}(\partial_{c}g^{-}_{bd} + \partial_{b}g^{-}_{cd} - \partial_{d}g^{-}_{bc})~,
\end{equation}
and
\begin{equation}
B^{a}_{bc} = \dfrac{\epsilon}{2}\tilde{g}^{ad}\left([g_{bd}]n_{c} + [g_{cd}]n_b -[ g_{bc}]n_d\right)~,
\end{equation}
where $\tilde{g}^{ad}$ is the \emph{average inverse metric}, defined by
\begin{equation}
\tilde{g}^{ad} := \dfrac{1}{2}(g^{ad+} + g^{ad-})~.
\end{equation}
A natural requirement of the thin-shell formalism is that the condition $B^{a}_{bc} \approx \displaystyle\lim_{\varepsilon\rightarrow 0}B^{a}_{\varepsilon bc} \approx 0$ be satisfied \cite{Huber2025}. In this work, however, we will not impose such a condition, since $[g_{ab}] \neq 0.$

The Riemann tensor can be calculated by considering Eq. (\ref{christoffel-generalized}), which results in
\begin{eqnarray}\label{riemann-generalized}
R^{a}_{\varepsilon bcd} \approx \theta_{\varepsilon}(l)R^{a+}_{bcd} + \theta_{\varepsilon}(-l)R^{a+}_{bcd} + (A^{a}_{bcd} + \tilde{A}^{a}_{bcd})\delta_{\varepsilon}(l) + Q^{a}_{bcd}\delta_{\varepsilon}'(l) + F^{a}_{bcd}\delta_{\varepsilon}^2(l),
\end{eqnarray}
where
\begin{equation}\label{Riemann-A1}
A^{a}_{bcd} = \epsilon\left([\Gamma^{a}_{bd}]n_c - [\Gamma^{a}_{bc}]n_d\right),
\end{equation}
\begin{eqnarray}\label{Riemann-A2}
\tilde{A}^{a}_{bcd} = \partial_{c}B^{a}_{bd} - \partial_{d}B^{a}_{cb} + \tilde{\Gamma}^{a}_{ce}B^{e}_{bd} + \tilde{\Gamma}^{e}_{db}B^{a}_{ce} - \tilde{\Gamma}^{a}_{de}B^{e}_{cb} - \tilde{\Gamma}^{e}_{cb}B^{a}_{de}~,
\end{eqnarray}
\begin{equation}\label{Riemann-Q}
Q^{a}_{bcd} =  \epsilon\left(B^{a}_{db}n_c - B^{a}_{cb}n_d\right)~,
\end{equation}
and
\begin{equation}\label{Riemann-F}
F^{a}_{bcd} = B^{a}_{ce}B^{e}_{db} - B^{a}_{de}B^{e}_{cb}~.
\end{equation}
In Eq.~(\ref{Riemann-A2}), $\tilde{\Gamma}^{a}_{ce}$ is the \textit{average Christoffel symbol}, given by
\begin{equation}\label{Christoffel-media}
\tilde{\Gamma}^{a}_{ce} := \dfrac{1}{2}\left(\Gamma^{a+}_{ce} + \Gamma^{a-}_{ce}\right)~.
\end{equation}

Note that in Eq.~(\ref{riemann-generalized}) \(R^{a+}_{bcd}\) and \(R^{a-}_{bcd}\) are the regular Riemann tensors on the \(l>0\) and \(l<0\) sides of \(\Sigma\), respectively, whereas the terms \(A^{a}_{bcd} + \tilde{A}^{a}_{bcd}\), \(Q^{a}_{bcd}\), and \(F^{a}_{bcd}\) are the singular contributions that capture the discontinuity. The term \(A^{a}_{bcd}\) is already present in the standard case, but all the others arise due to the metric discontinuity \( [g_{ab}] \neq 0 \).

The Ricci tensor is determined by the contraction $R_{bd} = R^{a}_{bcd},$ which allows us to write
\begin{eqnarray}\label{ricci-generalized}
R_{\varepsilon bd} \approx \theta_{\varepsilon}(l)R^{+}_{bd} + \theta_{\varepsilon}(-l)R^{-}_{bd} + (A_{bd} + \tilde{A}_{bd})\delta_{\varepsilon}(l) + Q_{bd}\delta_{\varepsilon}'(l) + F_{bd}\delta_{\varepsilon}^2(l)~.
\end{eqnarray}
The Ricci scalar is obtained by the contraction $R_{\varepsilon} = g^{bd}_{\varepsilon}R_{bd}.$ Here we need to be careful, since
\begin{eqnarray}
R_{\varepsilon} &\approx& g^{bd}_{\varepsilon}R_{bd} \nonumber \\ 
&\approx& (\theta_\varepsilon(l) g^{bd+} + \theta_\varepsilon(-l) g^{bd-}) \left( \theta_{\varepsilon}(l)R^{+}_{bd} + \theta_{\varepsilon}(-l)R^{-}_{bd} + \bar{A}_{bd}\delta_{\varepsilon}(l) + Q_{bd}\delta_{\varepsilon}'(l) + F_{bd}\delta_{\varepsilon}^2(l)\right)~, \nonumber
\end{eqnarray}
where $\bar{A}_{bd} = A_{bd} + \tilde{A}_{bd}$. Carrying out the products inside the parentheses and taking into account the multiplications of distributions, we obtain
\begin{equation}\label{escalar-ricci-generalized}
R_{\varepsilon} \approx   \theta_{\varepsilon}(l)R^{+} + \theta_{\varepsilon}(-l)R^{-} + \tilde{g}^{bd}\bar{A}_{bd}\delta_{\varepsilon}(l) + \tilde{g}^{bd}Q_{bd}\delta_{\varepsilon}'(l) + \left( \tilde{g}^{bd}F_{bd} - 2\tilde{g}^{bd}Q_{bd}\right)\delta_{\varepsilon}^2(l)~, 
\end{equation}
which yields the final regularized form of the Ricci scalar in the context of generalized functions. Note that $\tilde{g}^{bd}$ is the averaged inverse metric, which prevents the usual index–contraction manipulations, so that
\begin{equation}
\tilde{g}^{bd}\bar{A}_{bd} \neq \bar{A}, \tilde{g}^{bd}Q_{bd}\neq Q ~~\text{e}~~ \tilde{g}^{bd}F_{bd}\neq F~.
\end{equation}
By substituting the results (\ref{metricaregularizada}), (\ref{ricci-generalized}), and (\ref{escalar-ricci-generalized}), the Einstein field equations can be formulated in a generalized sense. Consequently, the Einstein tensor also decomposes into five parts, namely,
\begin{equation}
R_{\varepsilon ab} - \dfrac{1}{2}g_{\varepsilon ab}R_{\varepsilon} \approx 8\pi T_{\varepsilon ab}
\end{equation}
becomes
\begin{eqnarray}\label{einstein-generalized}
\theta_{\varepsilon}(l)\left(R^{+}_{\varepsilon ab} - \dfrac{1}{2}g^{+}_{\varepsilon ab}R^{+}_{\varepsilon}\right) + \theta_{\varepsilon}(-l)\left(R^{-}_{\varepsilon ab} - \dfrac{1}{2}g^{-}_{\varepsilon ab}R^{-}_{\varepsilon}\right) + \left(\bar{A}_{ab} - \dfrac{1}{2}\tilde{g}_{ab}\tilde{g}^{ab}\bar{A}_{ab}\right)\delta_{\varepsilon}(l) \nonumber \\
+ \left(Q_{ab} - \dfrac{1}{2}\tilde{g}_{ab}\tilde{g}^{ab}Q_{ab}\right)\delta'_{\varepsilon}(l) + \left(F_{ab} - \dfrac{1}{2}\tilde{g}_{ab}\tilde{g}^{ab}(F_{ab} - 4Q_{ab})\right)\delta^2_{\varepsilon}(l) \approx 8\pi T_{\varepsilon ab}~.
\end{eqnarray}
Observing the structure of (\ref{einstein-generalized}), we can express the energy–momentum tensor in the following form
\begin{equation}\label{tensor-EM-completo}
T_{\varepsilon ab} \approx \theta_{\varepsilon}(l)T^{+}_{ab} + \theta_{\varepsilon}(-l)T^{-}_{ab} + \delta_{\varepsilon}(l)S_{ab} + \delta'_{\varepsilon}(l)W_{ab} + \delta^2_{\varepsilon}(l)J_{ab}
\end{equation}
where
\begin{equation}
\lim_{\varepsilon \to 0}T_{\varepsilon ab} \approx T_{ab} \approx  \theta(l)T^{+}_{ab} + \theta(-l)T^{-}_{ab} + \delta(l) S_{ab} + \delta'_{\varepsilon}(l)W_{ab} + \delta^2_{\varepsilon}(l)J_{ab},
\end{equation}
with
\begin{equation}\label{tensor-S-generalized}
S_{ab} = \bar{A}_{ab} - \dfrac{1}{2}\tilde{g}_{ab}\tilde{g}^{ab}\bar{A}_{ab}~,
\end{equation}
\begin{equation}\label{tensor-W-generalized}
W_{ab} = Q_{ab} - \dfrac{1}{2}\tilde{g}_{ab}\tilde{g}^{ab}Q_{ab}~,
\end{equation}
and
\begin{equation}\label{tensor-J-generalized}
J_{ab} = F_{ab} - \dfrac{1}{2}\tilde{g}_{ab}\tilde{g}^{ab}(F_{ab} - 4Q_{ab})~.
\end{equation}
Equation~(\ref{tensor-EM-completo}) describes how the energy--momentum tensor \(T_{ab}\) behaves as one crosses the hypersurface \(\Sigma\). On each side of \(\Sigma\), the tensor \(T_{ab}\) is given by \(T^{+}_{ab}\) (for \(l > 0\)) and \(T^{-}_{ab}\) (for \(l < 0\)), which determine the local geometry in the usual way through Einstein’s equations. However, imposing a discontinuous metric in the Einstein tensor across \(\Sigma\) implies that Einstein’s equations cannot be satisfied solely by \(T^{+}_{ab}\) and \(T^{-}_{ab}\), nor simply by
\begin{equation}\label{tensorEMcompleto}
T_{\alpha\beta} = \theta(l)T^{+}_{\alpha\beta} + \theta(-l)T^{-}_{\alpha\beta} + \delta(l)S_{\alpha\beta}~.
\end{equation}
as in the standard case.

Therefore, we observe the appearance of additional terms, \(W_{ab}\) and \(J_{ab}\), together with \(S_{ab}\), which is already present in the standard Darmois–Israel conditions. Equation~(\ref{tensor-S-generalized}) is the analogue of
\begin{equation}\label{tensor-S-padrão}
S_{\alpha\beta} = \dfrac{1}{8\pi}\left(A_{\alpha\beta} - \dfrac{1}{2}Ag_{\alpha\beta}\right)
\end{equation}
and represents a singular contribution (proportional to \(\delta(l)\)) supported on the hypersurface \(\Sigma\). This term acts as an additional energy–momentum distribution on \(\Sigma\), ensuring that Einstein’s equations remain consistent throughout the spacetime, including across the hypersurface.

The challenge lies in providing a physical interpretation for Eqs.~(\ref{tensor-W-generalized}) and (\ref{tensor-J-generalized}). These terms do not make sense within the classical theory of distributions (since one cannot multiply arbitrary distributions), but within Colombeau algebras they are well defined. The Einstein equations obtained here reveal new generalized terms associated with the geometric discontinuity of the metric. In what follows, we discuss how these terms manifest in physical contexts and their possible interpretations.

\subsection{Possible Physical Interpretations of Singular Terms on \(\Sigma\)}

The classical formulation of junction conditions in general relativity, as introduced by Darmois and refined by Israel, relies on the assumption that the metric tensor is continuous across the hypersurface \(\Sigma\), that is, \([g_{ab}] = 0\). This assumption eliminates ambiguous products of distributions such as \(\theta(l)\delta(l)\) or \(\delta^{2}(l)\), ensuring that Einstein’s equations remain valid in the sense of classical distribution theory.

However, once the continuity condition on the metric is relaxed and we allow \([g_{ab}] \neq 0\), new singular terms emerge in the development of the field equations. These terms include not only the familiar \(\delta(l)\) contribution, which characterizes the presence of a thin shell of energy–momentum on \(\Sigma\), but also more singular contributions such as \(\delta'(l)\) and \(\delta^{2}(l)\).

Although such terms are not defined within the framework of Schwartz distributions, they become well defined in the context of Colombeau algebras of generalized functions. As presented in Section~\ref{sec2}, this formalism makes it possible to manipulate products and derivatives of distributions in a manner compatible with differential calculus, thereby providing a mathematically consistent treatment of the objects involved.

Within this framework, we now discuss possible physical interpretations that are compatible with the geometric nature of \(\Sigma\):

\begin{itemize}
	\item[(i)] The term proportional to \( \delta_{\varepsilon}(l) \) remains associated, as in the standard case, with an energy–momentum density concentrated infinitesimally on the hypersurface \( \Sigma \). This distributional contribution is widely recognized in the literature as representing a thin shell that separates two spacetime regions, with the localized energy and momentum on \( \Sigma \) determined by the tensor \( S_{ab} \).
	
	However, in the context of a discontinuous metric \( [g_{ab}] \neq 0 \), the structure of this term is fundamentally modified. The surface tensor \( S_{ab} \), obtained from the generalized Einstein equations, depends not only on the jump in the extrinsic curvature \( A_{ab} \), but also on an additional term \( \tilde{A}_{ab} \), which arises directly from the metric discontinuity.
	
	This modification shows that, although \( \delta_{\varepsilon}(l) \) still models an infinitesimally concentrated distribution, the physical content associated with this term is not the same as in the case of a continuous metric, since the surface energy density now also depends on the way in which the geometry breaks across \( \Sigma \). Its tensorial structure therefore reflects the direct influence of the metric discontinuity. Thus, in this new context, the term \( \delta_{\varepsilon}(l) \) captures not only the presence of a thin shell but also information about the internal geometry of the discontinuity.

	\item[(ii)] The term proportional to \( \delta'_{\varepsilon}(l) \) may be associated with abrupt variations in the geometry or in the energy distribution along the normal direction to the hypersurface \( \Sigma \). In the classical distributional framework, the derivative of the Dirac delta is only an operational object, without a pointwise definition and without the possibility of being multiplied with other distributions. In the context of Colombeau algebras, however, \( \delta'_{\varepsilon}(l) \) is a well-defined smoothened and concentrated function, which allows it to be interpreted as a marker of non-smooth transitions. Thus, this term may be linked to internal stresses, discontinuities in the derivative of the curvature, or to a variation in the energy–momentum distribution across \( \Sigma \), reflecting a more intricate internal structure of the layer (see Appendix~\ref{apend}).

	\item[(iii)] The term proportional to \( \delta_{\varepsilon}^2(l) \) represents an even higher level of singularity; unlike \( \delta_{\varepsilon}(l) \), it has no meaning within the classical Schwartz distribution framework. However, its properties, which parallel those of the Dirac delta, acquire a precise meaning within Colombeau algebras, where the quantity \( \int \delta^2(x)\, dx \) is interpreted as a generalized number of infinitely strong yet moderate nature. Typical regularizations of \( \delta_{\varepsilon}^2(l) \) reveal a highly localized structure, with a divergent maximum value and an integral that grows like \( \sim \varepsilon^{-1} \) in the regularization limit \( \varepsilon \to 0 \). This behavior suggests that \( \delta^2(l) \) carries a singular intensity more pronounced than that of a simple surface energy distribution (see Section~\ref{sec2}).

	In light of these properties, it is reasonable to interpret this term as the idealized manifestation of an ultra–concentrated energy or curvature density along the hypersurface \( \Sigma \). Such a structure could correspond, for instance, to a limiting regime of a thick layer with internally singular curvature, or to an extreme geometric transition. Works such as that of Huber~\cite{Huber2025} show that, when employing smoothed functions in related formalisms, objects with support spread over subregions of spacetime naturally arise, characterizing thick shells. In this sense, the presence of \( \delta_{\varepsilon}^2(l) \) may be interpreted as a mathematical marker of an extreme geometric transition or of an effectively thick layer. Although its definitive physical interpretation still requires further investigation, its formal consistency within Colombeau algebras justifies its inclusion as a legitimate component of the extended field equations.

\end{itemize}

It is important to emphasize that these interpretations are, for the time being, hypothetical. They are not grounded in experimental models or in any established physical consensus, but rather in the formal consistency of the Colombeau framework and in the mathematical structure that naturally emerges when \([g_{ab}] \neq 0\). As such, these ideas should be regarded as preliminary proposals, whose deeper exploration may open new perspectives on the nature of generalized functions in general relativity.

From this viewpoint, the new terms \( W_{ab} \) and \( J_{ab} \) may be understood as indicators of additional geometric degrees of freedom that arise when the continuity of the metric is relaxed. Although their definitive physical interpretation still requires further investigation, their presence signals a conceptually relevant step toward a more comprehensive description of singular structures in spacetime.

Assuming a discontinuous metric leads to a more general and intricate framework.  
Equations (\ref{tensor-S-generalized})–(\ref{tensor-J-generalized}) are not, in a strict sense, the Lanczos equations, since they are formulated as four-tensors. A complete formulation requires performing the projections onto the hypersurface, from which the three-tensors \(S_{\alpha\beta}\), \(W_{\alpha\beta}\), and \(J_{\alpha\beta}\) emerge. These objects are analogous to the classical Lanczos equation,  
\[
S_{\alpha\beta} = \frac{\epsilon}{8\pi}\left([K_{\alpha\beta}] - [K]\,h_{\alpha\beta}\right),
\]
but incorporate additional contributions produced by the discontinuity \( [g_{ab}] \neq 0 \).

In addition to the tangential projection, the normal components \(G_{ab}n^{a}n^{b}\) and the mixed components \(G_{ab}e^{a}_{\alpha}n^{b}\) play an essential role in characterizing how the discontinuity affects the flux of energy-momentum across \(\Sigma\). These terms are particularly relevant for \(W_{ab}\) and \(J_{ab}\), which, being more singular, may exhibit significant normal components and influence both the conservation of energy-momentum and the intrinsic and extrinsic dynamics of the hypersurface.

The geometric and physical impact of these projections will be examined in subsequent work, where a systematic analysis will clarify how the metric discontinuity modifies both the intrinsic structure of \(\Sigma\) and the behavior along the normal direction.

\section{Summary and Conclusions}\label{sec5}

In this work, we developed an extension of the Darmois-Israel formalism to the case in which the metric exhibits genuine discontinuities, that is, when $[g_{ab}] \neq 0$. The formulation was constructed systematically within the framework of Colombeau algebras, allowing a rigorous treatment of products and derivatives of distributions that arise in the Einstein field equations.

The resulting formalism provides a coherent mathematical structure for describing non-smooth junctions, preserving geometric consistency even in the presence of discontinuities in the metric and its derivatives. The analysis showed that a direct discontinuity in the metric generates a sequence of geometric effects: (i) the behavior of the normal vector and the induced metric requires additional consistency conditions, since only the covariant normal $n_i$ remains continuous, while $n^i$ and $h_{\alpha\beta}$ inherit the discontinuity of the inverse metric; (ii) the Christoffel symbols acquire terms proportional to the delta distribution and its derivative; (iii) the curvature tensor exhibits a decomposition by singularity order, with regular, singular, and ultrasingular parts; and (iv) in the generalized energy–momentum tensor, besides the standard thin-shell term, new contributions arise associated with the geometric breaking across the junction hypersurface, terms that do not appear in the classical Darmois-Israel case.

The proposed model recovers, as a special case, the classical Darmois-Israel conditions, but extends them to a more general setting in which spacetime may display abrupt transitions or complex internal structures. In this way, the Colombeau framework proves to be not only mathematically consistent but also conceptually robust, offering new avenues for the investigation of singularities, signature change, and geometric boundaries in spacetime.

\appendix
\section{Geometric and distributional properties of \(\delta'(l)\)} \label{apend}

In this appendix we present additional details concerning the tensor \( W_{ab} \), which is proportional to \( \delta'_{\varepsilon}(l) \). This allows a more refined analysis of the possible interpretations of this term.

Let \( l(x^a) \) be a scalar field on spacetime, and let \( \delta_{\varepsilon}(l) \) denote a regularization of the Dirac delta supported on the hypersurface
\[
\Sigma = \{ x^a \mid l(x^a) = 0 \}.
\]
Applying the chain rule, we obtain
\begin{equation}
\partial_a \delta(l(x)) = \delta'(l)\, \partial_a l(x).
\end{equation}
Since the normal vector is defined (up to sign) by \( n_a = \epsilon\, \partial_a l \), it follows that
\begin{equation}\label{derivada-delta}
\partial_a \delta(l) = \delta'(l)\, n_a .
\end{equation}

Equation~\eqref{derivada-delta} shows that:
\begin{itemize}
	\item The spacetime derivative of the delta distribution is still supported on \( \Sigma \); in this sense it ``lives'' on the hypersurface just as \( \delta(l) \) does.
	
	\item It carries a definite orientation, namely the direction of the normal vector \( n_a \).
	
	\item Thus, the derivative of a distribution concentrated on a hypersurface measures how that distribution changes when one moves away from \( \Sigma \) along the normal direction.
\end{itemize}

Therefore, \( \delta'(l) \) measures variation along the normal direction to \( \Sigma \). This makes it reasonable to view the tensor \( W_{ab} \) as indicating that the distribution of energy–momentum is not homogeneous or constant along \( l \), but exhibits internal variations. This interpretation does not yet have a fully established status in the literature, but we propose it as a plausible and mathematically well–founded viewpoint.

\end{document}